\pdfoutput=1

\RequirePackage{snapshot}

\RequirePackage{sty/bootstrap}
\PrependInputPath{sty,sty/img,fig}

\ProvideExpandableDocumentCommand \class {} {IEEEtranPlus}

\documentclass [a4paper,conference,compress,arxiv,final] \class

\usepackage{xparse}

\ProvideIf[false]{ieee}
\ProvideIf[false]{ieeepreview}
\ProvideIf[false]{blind}
\ProvideIf[false]{arxiv}
\ProvideIf[false]{compress}

\ProvideIf[false]{final}
\ProvideIf[true]{draft}
\IfClassLoadedWithOptionsTF \class {final} {\finaltrue\draftfalse} {}
\IfClassLoadedWithOptionsTF \class {ieee} {\ieeetrue} {}
\IfClassLoadedWithOptionsTF \class {arxiv} {\arxivtrue} {}
\IfClassLoadedWithOptionsTF \class {compress} {\compresstrue} {}

\ifdraft
\fi

\ifieee
  \PassOptionsToPackage{cmex10}{amsmath}
  \usepackage[hyperref=false,bibtex]{boilerplate}
\else
  \ifarxiv
    \usepackage[biblatex]{boilerplate}
    
  \else
    \usepackage[bibtex]{boilerplate}
  \fi
\fi

\ifnum 1=\ifieeepreview 1\else\ifieee 1\else 0\fi\fi
  \DeclareDocumentCommand{\orcidlink}{m}{}
\fi

\usepackage{multicol}
\usepackage{tabularray}

\usepackage[capitalize]{cleveref}
\crefname{equation}{}{}
\Crefname{equation}{Equation}{Equations}

\IfFileExists{short.bib}{
  \addbibresource{short.bib}
}{}
\IfFileExists{bib.bib}{
  \addbibresource{bib.bib}
}{}

\title{%
  Function Computation and Identification over \\ Locally Homomorphic Multiple-Access Channels
}
\date{\today}

\ifblind \author{}\else
\author{
  \IEEEauthorblockN{%
    Johannes~Rosenberger,\IEEEauthorrefmark{1}
    Holger~Boche,\IEEEauthorrefmark{1}
    Juan~A.~Cabrera,\IEEEauthorrefmark{2}
    Christian~Deppe\IEEEauthorrefmark{3}
  }
  \IEEEauthorblockA{%
    \{\tumail{johannes.rosenberger}%
      , \tumail{boche}%
    \}@tum.de,
    \uref{mailto:juan.cabrera@tu-dresden.de}{juan.cabrera@tu-dresden.de},
    \uref{mailto:christian.deppe@tu-braunschweig.de}{christian.deppe@tu-braunschweig.de}
  }
  \IEEEauthorblockA{%
    \IEEEauthorrefmark{1}
    TUM School of Computation, Information and Technology,
    Technical University of Munich,
  }
  \IEEEauthorblockA{%
    \IEEEauthorrefmark{2}
    Faculty of Electrical and Computer Engineering
    and Centre for Tactile Internet with Human-in-the-Loop (CeTI), \\
    Technische Universität Dresden
  }
  \IEEEauthorblockA{%
    \IEEEauthorrefmark{3}
    Institute for Communications Technology,
    Technische Universität Braunschweig
  }
}
\fi

\makeatletter
\hypersetup{
  pdftitle = \string\@title,
  pdfauthor = {Johannes Rosenberger, Holger Boche, Juan A. Cabrera, Christian Deppe},
  pdfkeywords = {},
  pdflang = {English},
}
\makeatother

\usepackage{tikz}
\usetikzlibrary{math}
\usetikzlibrary{fpu}
\pgfkeys{/pgf/fpu}
\pgfkeys{/pgf/fpu/output format=fixed}
\tikzmath{
  function fold(\p, \q) { return \p*(1-\q) + \q*(1-\p); };
  function Hbin(\p) {
    if \p == 0 || \p == 1 then {
      return 0;
    } else {
      return - \p * log2(\p) - (1 - \p)*log2(1-\p);
    };
  };
  function Ibsc(\p, \q) { return Hbin(fold(\p, \q)) - Hbin(\q)); };
  function Ibec(\p, \q) { return Hbin(\p) - \q * Hbin(\p); };
  function IgaussC(\s) { return log2(1 + \s); };
  function IgaussR(\s) { return IgaussC(\s)/2; };
  function IzChannel(\p, \d) { return Hbin(\p * \d) - \p * Hbin(\d); };
  function zChannelPopt(\d) {
    real \g;
    let \g = (1-\d)^((1-\d)/\d);
    return \g/(1 + \d*\g);
  };
  function qHthermal(\n) {%
    if \n then { % if \n != 0
      return (\n + 1) * log2(\n + 1) - \n * log2(\n);
    } else { %
      return 0;
    };
  };
}
\pgfkeys{/pgf/fpu=false}

\usetikzlibrary{intersections}

\NewDocumentCommand \tikzIdMacBipartiteRates {O{} m} {

\tikzmath{
  real \ticklen, \pcrossover, \capBSC;
  \ticklen = 0.03; % length of axis ticks
  \pcrossover = #2; % we are noiseless, for now
  \capBSC = Ibsc(0.5, \pcrossover);
}
\begin{tikzpicture}[auto,#1]
  \draw[->] (0,0) -- (0, {1+2*\ticklen}) node[above] {$R_2$};
  \draw[->] (0,0) -- ({1+2*\ticklen}, 0) node[right] {$R_1$};

  \draw (0, 1) -- +(-\ticklen, 0) node[left] {$1$};
  \draw (1, 0) -- +(0, -\ticklen) node[below] {$1$};
  \draw (0, 0.5) -- (-\ticklen, 0.5) node[left] {$0.5$};
  \draw (0.5, 0) -- (0.5, -\ticklen) node[below] {$0.5$};

  \draw[dotted, thin] (0, 0.5) -- (1, 0.5) (0.5, 0) -- (0.5, 1);
  \draw[helper/diagonal, name path=diagonal] (0, 0) -- (1, 1);

  \draw[rateBound/di-mac, name path=di-mac]
    (0, 1) -- (1,1) node[rateBound/di-mac/node] {deterministic ID-MAC} -- (1, 0);
    \draw[rateBound/transmission, name path=transmission]
    (0, \capBSC) -- (\capBSC, \capBSC) -- node[rateBound/transmission/node] {Transmission} 
    (\capBSC, 0);
  \fill[rateBound/di,name intersections={of=diagonal and di-mac}]
    (intersection-1) circle (2pt) node[rateBound/di/node] {deterministic ID};

  \fill[rateBound/di-t,name intersections={of=diagonal and transmission}]
    (intersection-1) circle (2pt) node[rateBound/di-t/node] {ID-T};

\end{tikzpicture}
}

\colorlet{identification}{blue}
\colorlet{transmission}{red}
\tikzset{
  font = {\sffamily\footnotesize},
  helper/diagonal/.style = dotted,
  rateBound/di-t/.style = {color=blue},
  rateBound/transmission/.style = {dashed,color=red},
  rateBound/di-mac/.style = {},
  rateBound/di/.style = {color=blue},
}
\IfFileExists{local-shorthands.tex}{
\NewDocumentCommand{\namedef}{om}{\emph{#2}\IfValueT{#1}{~(#1)}}

\NewExpandableDocumentCommand \defeq {} \coloneqq

\NewExpandableDocumentCommand \union {} \bigcup
\NewExpandableDocumentCommand \intersection {} \bigcap
\NewExpandableDocumentCommand \setand {} \cap
\NewExpandableDocumentCommand \setor {} \cup

\DeclareMathOperator \BSC {BSC}

\DeclareMathOperator \Enc {Enc}
\DeclareMathOperator \Dec {Dec}

\DeclareMathOperator \id {id}

\NewExpandableDocumentCommand \RanVars {} \bbV
\NewExpandableDocumentCommand \Probs {} \bbP
\NewExpandableDocumentCommand \Nats {} \bbN
\NewExpandableDocumentCommand \Reals {} \bbR
\NewExpandableDocumentCommand \Ints {} \bbZ
\NewExpandableDocumentCommand \Field {} \bbF

\NewExpandableDocumentCommand \compose {} \circ

\NewExpandableDocumentCommand \xto {} \xrightarrow

\NewExpandableDocumentCommand \grel {} {\sim_}

\NewExpandableDocumentCommand \ampersand {} {&}

\ProvideIf[false]{draft}
\ifdraft
  \let \fxcolor \color
  \let \fxtextcolor \textcolor
\else
  \NewExpandableDocumentCommand \fxcolor {m} {}
  \NewExpandableDocumentCommand \fxtextcolor {m} {}
\fi

}{}

\usepackage{simpleacros}
\ProvideIf[false]{compress}
\ifcompress
  \acsetup{
    first-style = long-short,
    subsequent-style = short,
  }
\else
  \acsetup{
    first-style = long,
    subsequent-style = long,
  }
\fi

\SimpleAcros{
  id = identification,
  bc = broadcast channel,
}
\DeclareAcronyms{
  lhc = {
    short = LHC,
    long = locally homomorphic channel,
    short-indefinite = an,
  },
}

\begin{document}

\maketitle

\ifonecol
  \input edit-notes
\fi

\begin{abstract}
  We develop the notion of a \acl{lhc}
  and prove an approximate equivalence between those and codes for computing
  functions. Further, we derive decomposition properties of \aclp{lhc}
  which we use to analyze and construct codes where two messages must be encoded
  independently. This leads to new results for \acl{id} and $K$-\acl{id}
  when all messages are sent over multiple-access channels,
  which yield surprising rate improvements compared to naive code constructions.
  In particular, we demonstrate that for the example of \acl{id}
  with deterministic encoders, both encoders can be constructed independently.
\end{abstract}

\section{Introduction}

The steady exponential growth of data usage in communication
poses a major problem for the design of communication networks: By Shannon's
famous capacity result~\cite{shannon1948it0} that pioneered the field
of information theory, data rates
are limited by the channel capacity, a fundamental property of the channel that
can only be increased, in the case of wireless and optical channels,
by increasing the bandwidth of the signal or the average transmission
power per symbol.
Thus, it is predicted that the generated data rates will exceed the
available channel capacities by around 2040~\cite{src2021decadalPlan},
as the capacity growth is limited by energy consumption
and electromagnetic compatibility.

Is there no escape from Shannon's limit? There is: He assumed that the sender
chooses a message from a finite set, encodes it, and the receiver's task is
to reproduce it exactly.
However, communication always has a context, and frequently, only
certain aspects of the message are important to the receiver.
Thus, one can model the communication goal as enabling the receiver to
compute a function of the context, the sender's, and the receiver's
knowledge, using a communication channel.
For example, consider a cyber-physical
system where Alice is an in-network controller accessing the state
$a \in \cA$ of
multiple sensors deployed in a factory. On the other hand,
Bob is a controlled plant that makes control decisions (i.e., compute a
function $f$) based on its own data $b \in \cB$, and Alice's messages.
This is considered in the paradigms of
semantic and goal-oriented communication~\cite{guenduez2023semantic}
to improve data compression,
often relying on
traditional rate-distortion theory~\cite{shannon1959rd_theory},
by using semantic and goal-oriented distortion functions and
side information
\cite{yamamoto1982wyner,alonOrlitsky1996sourceCodingGraph,orlitskyRoche2001coding_computing,doshi2010functional}.
Such so-called \emph{functional compression} has, e.g., been applied to
control problems \cite{rezwanCabreraFitzek2022funcomp}.
These more traditional approaches are used to compute a function multiple times
with an average distortion criterion, and they
enable compression rate improvements on an
exponential scale, except for functions with special symmetries like,
e.g., the modulo-two-sum of two numbers.

Similar ideas to exploit communication goals are the basis for Post-Shannon
theory~\cite{cabreraEA2021postShannon6G} and communication
complexity~\cite{yao1979comm_complexity,kushilevitzNisan1996comm_complexity}.
Pioneering work by Yao~\cite{yao1979comm_complexity}, JaJa~\cite{jaja1985identification}, and
Ahlswede and Dueck~\cite{ahlswedeDueck1989id1} found that for the task
of message \acf{id},
the efficiency is significantly improved by joint compression and error correction
\cite{jaja1985identification}, or by using stochastic
encoders~\cite{yao1979comm_complexity,ahlswedeDueck1989id1} to encode ambiguously.
In this setting, the sender and receiver select a message each, and the receiver
wants to test whether they are equal or not.
Randomized \ac{id} enables compression not only by some ratio, but to
logarithmic size, for a maximal error criterion.
Under a uniform averaging, code sizes are entirely unbounded
\cite{hanVerdu1992idNewResults}, even for deterministic (non-stochastic) codes
\cite{rosenbergerIbrahimDeppeFerrara2023di_mac_isit}.
Whereas \ac{id} is a very specific
task, its theory has been generalized
in part to selected other
problems~\cite{ahlswede2008gtit_updated} such as
$K$-\ac{id}, where the receiver infers if
the sent message equals one of $K$ many hypotheses.

We set out to find a more general characterization of similar communication
tasks where the use of randomness admits efficiency gains. To this end,
we took inspiration from describing functions as
graphs~\cite{witsenhausen1976characteristicGraph}.
We describe functions by hypergraphs that collect the maximal independent sets
of their characteristic graphs.
Since the task of a function computation code is to preserve the structure of
the function over a channel, we analyze homomorphic, i.e., structure-preserving,
maps between hypergraphs, instead of reasoning mainly about properties of the hypergraphs.
As noisy channels cannot be hypergraph homomorphisms,
we define \acfp{lhc} (\cref{sec:defs}). In \cref{sec:codesLhc}, we show an equivalence
between function codes and, and \acp{lhc}.
Further, we prove that for reliable computation of a function,
even the channel itself must be locally homomorphic. Hence, stochastic coding
can never help much, if encoders and decoders are unconstrained. When
messages are distributed among parties, e.g. one message is known to the
sender and one to the receiver, or both to different senders,
stochastic encoding can help indeed, and in \cref{sec:bipartite}, we
indicate how such bipartite encoders can be constructed from encoders
that assume the receiver knows the “other” message,
even if both messages are sent over a multiple-access channel.
For this setting, we construct \ac{id}-codes in \cref{sec:identification}.
We demonstrate our results at the example of deterministic \ac{id} over
a pair of binary symmetric channels, where the channel preserves the distance
of the channel inputs.
We compare the achievable rates and to simply using transmission
codes and computing the function at the channel output.
Finally, we discuss where tradeoffs between the rates for different partial
channels can be found.

\section{Definitions}
\label{sec:defs}

We use the following notation:
The
set of random variables over a finite set $\cX$ is denoted by
$\RanVars(\cX)$.
The indicator function $\ind{\cdot}$ evaluates to $1$
if its argument is true, and to $0$ if it is false.
For values $x, a$, we define $x \otimes a \coloneqq (x,a)$,
and for functions $f : \cX \to \cY$ and $g : \cA \to \cB$,
the function $f \otimes g : \cX \times \cA \to \cY \times \cB$ is defined by
$(f \otimes g) (x,a) \coloneqq (f(x), f(a))$.
We denote the image of $\cS \subseteq \cX$ under $f$ by
$f(\cS) \coloneqq \set{ f(x) : x \in \cS }$.

\subsection{Channels and codes}

\begin{definition}[(discrete memoryless) channel]
  A \namedef{channel} $\phi : \cX \to \RanVars(\cY)$ is a function that
  maps an input symbol $x \in \cX$ to a random variable $Y = \phi(x)$.
  A \namedef{discrete channel} has finite $\cX, \cY$.
  A \namedef{memoryless channel $\phi$} is a family $\set{ \phi^n }_{n \in \Nats}$,
  where $\phi$ is applied independently to each letter of a sequence $x^n \in
  \cX^n$, i.e., $\phi^n(x^n) = \brack{ \phi(x_1),\dots,\phi(x_n) }$.
\end{definition}

If the goal of communication is to compute a function of a pair of arguments
selected by sender and receiver, we equivalently want to transform the channel
into the function, with high reliability.
A code is used to translate between structures that the users understand
and a structure that is preserved with high probability by the channel.

\begin{definition}
  A \namedef{function code} $(\Enc, \Dec)$ for computing
  a function $f : \cA \to \cB$
  over a channel $\phi : \cX \to \RanVars(\cY)$
  consists of an encoder $\Enc : \cA \to \RanVars(\cX)$
  and a decoder $\Dec : \cY \to \RanVars(\cB)$.
  It is an \namedef{$(f, \phi, \bm\lambda)$-code}, where $\bm\lambda \in
  [0,1]^{\card{f(\cA)}}$, if for all $a \in \cA$,
  \begin{gather}
    \Pr\set{ f(a) = \Dec \compose\, \phi \compose \Enc \tup*a } \ge 1 - \lambda_{f(a)}
    \,.
  \end{gather}
\end{definition}

Thus, for any input $a$, the channel
$\Dec \compose\, \phi \compose \Enc$
preserves the structure of the function $f$,
with probability $\ge 1 - \lambda_{f(a)}$.
We can express more traditional communication codes as function codes. To this
end, let $M \in \bbN$ and $\cM \coloneqq \set{1,\dots,M}$.

\begin{definition}
  An
  $(M, \phi, \bm\lambda)$-\namedef{transmission-code}, also called a \namedef{channel code}
  is an $(\id_\cM, \phi, \bm\lambda)$-code, where $\id_\cM : m \mapsto m$ is the identity function on $\cM$.
  It must exactly reproduce any $m \in \cM$ at the receiver's side, with probability $\ge 1 - \lambda_m$.
\end{definition}

\begin{definition}
  To define $K$-\ac{id}-codes, let $\cM^K \coloneqq \set{ \cS \subseteq \cM : \card\cS = K }$.
  An $(M, K, \phi, \bm\lambda)$-\namedef{\acl{id}-code}
  is an $(f_{K\textsf{-ID}}, \phi, \bm\lambda)$-code where
  $f_{K\textsf{-ID}} : \cM \times \cM^K \to \set{ 0, 1 },\, (m, \cS) \mapsto \ind{ m \in \cS }$.
  An $(M, \phi, \bm\lambda)$-\namedef{\ac{id}-code} is an $(M, 1, \phi, \bm\lambda)$-\ac{id}-code.
  Equivalently, we can express \ac{id} by the function
  $f_{\textsf{ID}} : \cM^2 \to \set{0,1},\, (m, m') \mapsto \ind{ m = m' } \equiv
  f_{1\textsf{-ID}}$.
\end{definition}

\subsection{Modeling functions as characteristic hypergraphs}

To compute a function over a channel, we must preserve its structure by means of coding.
If the structure of $f$ is expressed by a hypergraph, i.e., a set equipped with a
set of subsets, then a structure-preserving map is a hypergraph
\namedef{homomorphism}. To incorporate noise and random resources
in our model, we define \namedef{\aclp{lhc}}.

\begin{definition}
  A \namedef{hypergraph} $G = (\cV, \cE)$ consists of a vertex
  set $\cV(G) \coloneqq \cV$ and an edge set $\cE(G) \coloneqq \cE$,
  where every edge $e \in \cE$ is a subset of $\cV$.
  In the \namedef{complete 1-uniform hypergraph}
  $\cV^1 := \tup{\cV, \set{ \set{v} : v \in \cV}}$,
  all vertices are disconnected.
  If the edges partition the vertex set, $G$ is called
  a \namedef{partition hypergraph}.
  A \namedef{graph} is a 2-uniform hypergraph ($\cE \subseteq \cV^2$).
\end{definition}

\begin{definition}
  The \namedef{characteristic hypergraph} $H_f = (\cA, \cE_f)$ of a function
  $f : \cA \to \cB$ has $\cE_f = \set{ \set{ a : f(a) = b } : b \in f(\cA)}$.
\end{definition}

\begin{definition}
  A \namedef{homomorphism} $f : G \to H$ from a hypergraph $G$ to a hypergraph $H$
  is a function $f : \cV(G) \to \cV(H)$ that preserves edges,
  i.e., there exists a function $f_\cE : \cE(G) \to \cE(H)$ such that
  for all $e \in \cE(G)$, $f(e) \subseteq f_\cE(e)$. The function $f_\cE$ is
  called the edge map of $f$.
  If $f_\cE$ is sur-/bijective, $f$ is called
  \namedef{edge-sur/-bijective}, respectively.
\end{definition}
\begin{remark}
  Note that every function $f : \cA \to \cC$ is an edge-bijective homomorphism $f : H_f \to \cC^1$.
  Further, if $f$ is a homomorphism $G \to H$, then
  all intersecting pairs of edges $e, e' \in \cE(G)$ have
  $f(e \cap e') \subseteq f_\cE(e) \cap f_\cE(e')$ and $f(e \cup e') \subseteq
  f_\cE(e) \cup f_\cE(e')$.
  Edge-surjectivity of $f$ means equivalently that $f_\cE(\cE(G)) = \cE(H)$,
  and edge-injectivity means that $f_\cE^{-1}(\cE(H)) = \cE(G)$.
  Hence, if $f$ is edge-surjective and
  a pair $f_\cE(e) \cap f_\cE(e') = \emptyset$, then $e \cap e' = \emptyset$.
  For example, for all $c,c' \in \cE(\cC^1)$, we have that $c = c'$ or
  $c \cap c' = \emptyset$, and thus for all $e, e' \in \cE(H_f)$, $e = e'$
  or $e \cap e' = \emptyset$ .
\end{remark}

\begin{definition}[\acf{lhc}]
  \label{def:lhom}
  Let $G$ and $H$ be hypergraphs and $\bm\lambda \in [0,1]^{\card{\cE(G)}}$.
  A channel $\phi : \cV(G) \to \RanVars(\cV(H))$ is called a
  \namedef{\acl{lhc}} $\phi : G \xto{\bm\lambda} H$
  if there exists a function $f_\cE : \cE(G) \to \cE(H)$,
  called the edge map of $\phi$,
  such that for every $a \in A \in \cE(G)$,
  \begin{gather}
    \label{eq:def.lhom}
    \Pr\set[\Big]{ \phi(a) \in \bigcap_{\substack{A \in \cE(G) \\ a \in A}} f_\cE(A) }
    \ge 1 - \min_{\substack{A \in \cE(G) \\ a \in A}} \lambda_A
    \,.
  \end{gather}
  If $f_\cE$ is sur-/bijective, $\phi$ is
  \namedef{edge-sur-/bijective}, respectively.
\end{definition}
\begin{remark}
  \label{remark:def.lhom.partitionHypergraph}
  If $G$ is a partition hypergraph, we denote by $G_a$
  the unique edge of $G$ containing $a \in \cV(G)$.
  Thus, \cref{eq:def.lhom} simplifies to
  \begin{gather}
    \label{eq:def.lhom.partitionHypergraph}
    \Pr\set{ \phi(a) \in f_\cE(G_a) }
    \ge 1 - \lambda_{G_a}
    \,.
  \end{gather}
\end{remark}

\section{Codes and \aclp{lhc}}
\label{sec:codesLhc}

For functions with only one argument, known to the sender,
the most communication-efficient way to
compute them over a channel is to simply send the function value.
However, for functions $f : \cA \times \cB \to \cC$
where both arguments are known only to different parties
which can only communicate over a channel, this is not possible.
In the following, we show that function codes give rise to
edge-bijective \ac{lhc} between partition hypergraphs, and vice versa.
Further, we examine decomposition properties which we use
in \cref{sec:bipartite} to analyze and construct bipartite encoders.
This enables us to obtain new results for ($K$)-\ac{id}
in \cref{sec:identification}.

\begin{lemma}
  \label{lemma:edgeMapHasHom}
  Any $f_\cE : \cE(G) \to \cE(H)$, where $G, H$ are
  partition hypergraphs,
  is the edge map of a homomorphism $f : G \to H$.
\end{lemma}
\begin{proof}
  For every $a \in \cV(G)$,
  let $G_a \in \cE(G)$ be the unique edge
  containing $a$, and select arbitrarily an image $f(a) \in f_\cE(G_a)$.
  Thus, for every edge $e \in \cE(G)$, $f(e) \subseteq f_\cE(e)$.
\end{proof}

\begin{lemma}
  \label{lemma:homBijectiveEquiv}
  If $\phi, \eta : H \xto{\bm\lambda} G$ are \acp{lhc} and $\eta$ is edge-bijective,
  there exists an edge-bijective homomorphism $g : G \to G$ such that
  for every $E \in \cE(H)$, $g \compose \eta(E) = \phi(E)$.
\end{lemma}
\begin{proof}
  Let $f_\cE, h_\cE$ be the edge maps of $\phi, \eta$, respectively.
  By bijectivity of $h_\cE$, we can define the function $g_\cE = f_\cE \compose h_\cE^{-1}$.
  By \cref{lemma:edgeMapHasHom}, there exists $g$ with edge map $g_\cE$,
  and hence, the edge map of $g \compose \eta$ is given by
  $g_\cE \compose h_\cE = f_\cE$.
\end{proof}

\begin{theorem}
  \label{lemma:fcodeMakesLhomChannel}
  For a function $f : \cA \to \cB$, the following holds:
  \begin{subthms}
    \subthm \label{lemma:fcodeMakesLhomChannel.homToCode}
      If $\Dec \compose\, \phi \compose \Enc : H_f \to f(\cA)^1$ is an
      edge-bijective \ac{lhc},
      there exists an edge-bijective homomorphism $g : f(\cA)^1 \to f(\cA)^1$
      such that $(\Enc, g \compose \Dec)$ is an $(f, \phi, \bm\lambda)$-code.
    \subthm \label{lemma:fcodeMakesLhomChannel.codeToHom}
      If $(\Enc, \Dec)$ is an $(f, \phi, \bm\lambda)$-code, then
      $\Dec \compose\, \phi \compose \Enc : H_f \to f(\cA)^1$ is an
      edge-bijective \ac{lhc}
  \end{subthms}
\end{theorem}
\begin{proof}
  Let $\psi = \Dec \compose\, \phi \compose \Enc$.
  By the definitions of either a function code or \iac{lhc},
  \begin{align}
    1 - \lambda_{f(a)}
      &\le \Pr\set{ \psi(a) = f(a) }
    \\&= \Pr\set{ \psi(a) \in f_\cE(H_{f,a}) }
    \,,
  \end{align}
  and $f$ is an edge-bijective homomorphism $H_f \to f(\cA)^1$.
  Hence, \cref{lemma:fcodeMakesLhomChannel.homToCode} holds by \cref{lemma:homBijectiveEquiv},
  and \cref{lemma:fcodeMakesLhomChannel.codeToHom} holds
  because $f_\cE$ is bijective, and thus $\psi$ is also
  edge-bijective.
\end{proof}

\subsection{Decomposing \aclp{lhc}}

\begin{lemma}
  \label{lemma:detPrefix_lhomSuffix}
  Consider any hypergraphs $F, G, H, I$,
  edge-bijective homomorphisms $f : F \to G$, $h : H \to I$,
  and channels $\gamma : \cV(G) \to \RanVars(\cV(H))$.
  $h \compose \gamma \compose f : F \xto{\bm\lambda} I$
  is \iac{lhc} with edge map $e_\cE$
  if and only if $\gamma : G \xto{\bm\lambda} H$
  is locally homomorphic
  with edge map $g_\cE = h_\cE^{-1} \compose e_\cE \compose f_\cE^{-1}$.
\end{lemma}
\begin{proof}
  For all $a,b$ that satisfy $b = f(a)$,
  \begin{align}
    &1 - \min\nolimits_{A \in \cE(F),\, a \in A} \lambda_A
    \nonumber\\
      &\le \Pr\set{ \forall\, A \in \cE(F), a \in A : h \compose \gamma \compose f (a) \in e_\cE (A) }
    \\
      &=   \Pr\set{ \forall\, A \in \cE(F), a \in A : \gamma \compose f (a) \in g_\cE \compose f_\cE (A) }
    \\
      &= \Pr\set{ \forall\, B \in \cE(G), b \in B : \gamma(b) \in g_\cE(B) }
    \,,
  \end{align}
  where the first equality follows from edge-bijectivity of $h$
  and the second equality follows from edge-bijectivity of $f$, i.e.,
  for all $A \in \cE(F)$ and $C \in \cE(H)$,
  $a \in A \Leftrightarrow f(a) \in f_\cE(A)$ and
  $c \in C \Leftrightarrow h(c) \in h_\cE(C)$.
\end{proof}

\begin{theorem}
  \label{lemma:lhomChannel.decompose}
  Consider partition hypergraphs $H, F$
  and channels $\phi : \cA \to \RanVars(\cB)$ and $\gamma : \cB \to \RanVars(\cC)$ such that
  $\eta = \gamma \circ \phi : H \xto{\bm\lambda} F$ is an edge-bijective \ac{lhc}
  with edge map $e_\cE$, and $\lambda_A < 1/2$ for every $A \in \cE(H)$.
  Then, $\phi : H \xto{\bm\mu} G$ and $\gamma : G \xto{\bm\kappa} F$
  are edge-bijective \acp{lhc} whenever all of the following holds:
  \begin{subthms}
    \item $\cV(G) = \bigcup_{B \in \cE(G)} B$,
    \item $\cE(G) = \set{ B_A : A \in \cE(H) })$,
    \item $B_A = \set{ b : \Pr(\gamma(b) \in e_\cE(A)) > 1 - \kappa_{B_A} }$,
    \item $\lambda_A/\mu_A \le \kappa_{B_A} \le \frac 1 2$.
  \end{subthms}
\end{theorem}
\begin{proof}
  Let $W$ and $V$ be the conditional PMFs governing the channels $\phi$ and $\gamma$, respectively.
  For any $a \in \cA$, let $A = H_a \in \cE(H)$, $b \in B_A$, $C_A = e_\cE(A)$,
  and let $f_\cE : A \mapsto B_A$.
  Thus, $f_\cE(A) = \set{ b' : V_{b'}(C_A) \ge 1 - \kappa_{B_A} }$.
  Since $\eta$ is \iac{lhc},
  \begin{align}
    \mu_A \ge \lambda_A/\kappa_{B_A}
      &\ge \kappa_{B_A}^{-1} \Pr\set{ \eta(a) \notin e_\cE(A) }
    \\&=   \kappa_{B_A}^{-1} \sum\nolimits_{b'} W_a(b') V_{b'}(C_A^c)
    \\&=   \expect\brack{ V_{\phi(a)}(C_A^c) } /  \kappa_{B_A}
    \\&\ge W_a\tup[\Big]{ V_{\phi(a)} (C_A^c) \ge \kappa_{B_A} }
    \\&=   \Pr\set{ \phi(a) \notin f_\cE(H_a) }
    \,.
  \end{align}
  Since every $\kappa_{B_A} \le 1/2$, and all $C_A$ are pairwise disjoint,
  there is only one $A$ where
  $\Pr\set{ \gamma(b) \in C_A } > 1 - \kappa_{B_A}$.
  Hence, $f_\cE$ is bijective, and $\phi : F \xto{\bm\mu} G$ is an
  edge-bijective \ac{lhc}.
  Further, let $g_\cE = e_\cE \compose f_\cE^{-1}$ be a bijection.
  Now, for every $b \in B_A$,
  \begin{align}
    1 - \kappa_{B_A}
      &< \Pr\set{ \gamma(b) \in e_\cE(A) }
    \\&= \Pr\set{ \gamma(b) \in e_\cE \compose f_\cE^{-1}(B_A) }
    \\&= \Pr\set{ \gamma(b) \in g_\cE(G_b) }
    \,.
  \end{align}
  Therefore, $\gamma : G \xto{\bm\kappa} H$ is an edge-bijective \ac{lhc}.
\end{proof}
By the following corollaries,
the channel itself is a reliable code somehow,
and decoding gains from randomness are small.

\begin{corollary}
  \label{corollary:fcode.ran.channelIsLhom}
  For any function $f : \cA \to \cB$ and channel $\phi : \cX \to \RanVars(\cY)$,
  if there exists an $(f, \phi, \bm\lambda)$-code
  and $4\bm\lambda \le \bm\kappa \le \frac 1 2$,
  then there exist partition hypergraphs $G$ and $F$,
  $\card{\cE(G)} = \card{\cE(F)} = \card{f(\cA)}$,
  and $\phi : G \xto{\bm\kappa} F$ is an edge-bijective \ac{lhc}.
\end{corollary}
\begin{proof}
  Suppose the code is $(\Enc, \Dec)$ and denote $\psi = \phi \compose \Enc$.
  By \cref{lemma:fcodeMakesLhomChannel},
  $\Dec \compose \psi : H_f \xto{\bm\lambda} f(\cA)^1$
  is an edge-bijective \ac{lhc}.
  By \cref{lemma:lhomChannel.decompose} applied to the channel $\Dec \compose\, \psi$,
  there exists a partition hypergraph $F$ such that
  $\phi \compose \Enc : H_f \xto{2 \bm\lambda} F$ an edge-bijective \ac{lhc},
  and $\card{\cE(F)} = \card{\cE(H_f)} = \card{f(\cA)}$.
  By applying \cref{lemma:lhomChannel.decompose} a second time,
  but to the channel $\phi \compose \Enc = \psi$,
  we have that for every $\bm\kappa$, $\frac 1 2 \ge \bm\kappa \ge 4 \bm\lambda$,
  there exists a partition hypergraph $G$ such that $\phi : G \xto{\bm\kappa} F$
  is an edge-bijective \ac{lhc},
  and $\card{\cE(G)} = \card{\cE(F)} = \card{f(\cA)}$.
\end{proof}

By \cref{corollary:fcode.ran.channelIsLhom,lemma:fcodeMakesLhomChannel},
stochastic encoding and decoding can at best halve the error probability:
\begin{corollary}
  \label{corollary:detDecoder}
  For any $(f, \phi, \bm\lambda)$-code $(\Enc, \Dec)$, there exist
  deterministic functions (homomorphisms) $\Enc', \Dec'$
  such that $(\Enc', \Dec')$ is an $(f, \phi, 4\bm\lambda)$-code.
\end{corollary}

\subsection{Bipartite encoders}
\label{sec:bipartite}

By \cref{corollary:detDecoder}, deterministic encoders are sufficient
if doubling the error probability is acceptable.
But for \ac{id}, stochastic encoding does indeed improve the
rates, when the receiver knows one message.
This holds because the messages must be encoded independently,
i.e., the encoder is a bipartite channel $\Enc = \Enc_1 \otimes \Enc_2$.
Still, we can replace the encoder by a deterministic one,
by~\cref{lemma:lhomChannel.decompose,corollary:detDecoder,corollary:fcode.ran.channelIsLhom},
but the new encoder may not be bipartite.
We construct bipartite encoders from codes that assume the “other” message
were known to the receiver.

\begin{corollary}
  \label{corollary:lhomChannel.bipartite.semiDet}
  Consider $H, F, \phi, \bm\mu$ as in \cref{lemma:lhomChannel.decompose} such that
  $\cA = \cA_1 \times \cA_2$, $\cB = \cB_1 \times \cB_2$, and
  $\phi = \phi_1 \otimes \phi_2 : \cA \to \RanVars(\cB)$.
  There exist partition hypergraphs $G_1, G_2$, where
  $\cV(G_1) \subseteq \cA_1 \times \cB_2$ and $\cV(G_2) \subseteq \cB_1 \times \cA_2$,
  such that
  $\id_{\cA_1} \otimes\, \phi_2 : H \xto {\bm\mu} G_1$
  and $\phi_1 \otimes \id_{\cA_2} : H \xto {\bm\mu} G_2$
  edge-bijective \acp{lhc}
  and $\card{\cE(G_1)} = \card{\cE(G_2)} = \card{\cE(F)}$.
\end{corollary}

\Cref{corollary:lhomChannel.bipartite.semiDet} shows that if the concatenation
of encoder and channel is split into parallel parts (e.g., multiple
encoders and parallel channels), then the optimal encoder of one argument
does not depend on the encoders of the other arguments. In particular, for
noiseless channels, it does not matter to User 1 if User 2 is separate from or
identical to the receiver.

\begin{example}[($K$)-\acl{id}]
  For ($K$)-\ac{id} when both $m$ and $\cS$ are sent over parallel noisy or
  noiseless channels, it has been unclear if we must transmit exactly
  one argument, we can compress both independently,
  or whether there is a tradeoff between compression of the two.
  \Cref{corollary:lhomChannel.bipartite.semiDet}
  shows the optimal compression of both arguments $m$ and $\cS$ is independent.
\end{example}

The next lemma is key to constructing bipartite encoders from encoders
for partial identity channels.
Note that it requires a hypergraph $F$ to have a rectangular vertex set.

\begin{lemma}
  \label{lemma:bipartite.lhomChannel.changeIdBranch}
  For all partition hypergraphs $I, H, G, F$ with equal number of edges and
  \begin{align*}
    \cV(I) &= \cX_1 \times \cA_2 \,,
           &\cV(G) &= \cA_1 \times \cX_2
           \,, \\
    \cV(H) &= \cA_1 \times \cA_2 \,,
           &\cV(F) &= \cX_1 \times \cX_2
           \,,
  \end{align*}
  and for every edge-bijective \ac{lhc} $\id_{\cA_1} \otimes\, \phi : H \xto{\bm\lambda} G$,
  \begin{gather}
    \id_{\cX_1} \otimes\, \phi : I \xto{\bm\lambda} F
  \end{gather}
  is an edge-bijective \ac{lhc}.
\end{lemma}
\begin{proof}
  Let $t_\cE$ be the edge map of $\id_{\cA_1} \otimes \phi$ and
  denote by $F_{x_1,x_2}$ the edge of $F$ that contains $(x_1,x_2)$.
  There exist bijections $g_\cE, r_\cE, s_\cE, r_\cE$ such that
  \begin{gather}
    \cE(I) \xto[h_\cE]{} \cE(H)
      \xto[t_\cE]{} \cE(G)
      \xto[r_\cE]{} \cE(F)
      \xto[g_\cE]{} \cE(I)
    \,,
  \end{gather}
  $r_\cE = s_\cE^{-1}$,
  and $t_\cE \compose h_\cE \compose g_\cE (F_{x_1x_2}) = s_\cE(F_{x_1x_2})$.
  Let $h : I \to H$, $h_1 : (x_1, a_2) \mapsto \tup{ h_1(x_1,a_2), a_2) }$
  be a homomorphism with edge map $h_\cE$, i.e., $h(x_1,a_2) \in h_\cE(I_{x_1a_2})$.
  Define similarly $g : F \to I$, $s : F \to G$ and $r : G \to F$.
  Thus, if
  \begin{align}
    (\id_{\cA_1} \otimes \phi) \compose h (x_1, a_2)
      &=  (h_1(x_1,a_2), \phi(a_2))
    \nonumber\\
      &\in G_{h_1(x_1,a_2), \phi(a_2)}
    \\&= t_\cE \compose h_\cE(I_{x_1,a_2})
    \\&= s_\cE (F_{x_1,\phi(a_2)})
    \,,
  \shortintertext{then}
    s \compose (\id_{\cX_1} \otimes \phi) (x_1, a_2)
      &=  s(x_1,\phi(a_2))
    \nonumber\\
      &\in
    s_\cE(F_{x_1,\phi(a_2)})
    \\&= G_{h_1(x_1,a_2), \phi(a_2)}
    \,.
  \end{align}
  Hence, since $\id_{\cA_1} \otimes \phi : H \xto{\bm\lambda}$
  is an edge-bijective \ac{lhc},
  $(\id_{\cA_1} \otimes \phi) \compose h : I \xto{\bm\lambda} G$
  and $s \compose (\id_{\cX_1} \otimes \phi) : I \xto{\bm\lambda} G$
  are edge-bijective \acp{lhc} as well, and the same holds
  by \cref{lemma:detPrefix_lhomSuffix} for
  $
    (r \compose s) \compose (\id_{\cX_1} \otimes\, \phi) : I \xto{\bm\lambda} F
    $
    and $
    \id_{\cX_1} \otimes\, \phi\:\: : I \xto{\bm\lambda} F
    $.
\end{proof}

\subsection{\Acl{id}}
\label{sec:identification}

With the help of \cref{lemma:lhomChannel.decompose}
and \cref{lemma:bipartite.lhomChannel.changeIdBranch},
we can make surprising predictions
about \ac{id} with a bipartite encoder, i.e., $\Enc = \eta_1 \otimes
\eta_2$, which are not obvious from codes for the usual \ac{id} setting
where one $m' \in \cM$ is known to the receiver, i.e., we have a
channel is $\phi = \psi \otimes \id_\cM$.
In general, for $\phi : \cX_1 \times \cX_2 \to \RanVars(\cY)$, we can write
\begin{align}
  \Dec \compose\, \phi \compose \Enc
    &= \Dec \compose\, \phi
      \compose (\id_{\cX_1} \otimes \eta_2)
      \compose (\eta_1 \otimes \id_\cM)
    \,.
\end{align}
By \cref{lemma:lhomChannel.decompose} and
\cref{corollary:lhomChannel.bipartite.semiDet},
there exist partition hypergraphs $G_1, G_2, F$
such that we have the locally homomorphic chain
\begin{gather}
  H_f \xto[\eta_1 \otimes \id_\cM]{\bm\alpha}
  G_1 \xto[\id_{\cX_1} \otimes \eta_2]{\bm\beta}
  F \xto[\Dec \compose\, \phi]{\bm\mu}
  \set{0,1}^1
  \,,
  \label{eq:idCode.decompose}
\end{gather}
where $\bm\alpha \bm\beta \bm\mu = \bm\lambda$, if $\bm\lambda < 1/2$.
Furthermore, by \cref{lemma:bipartite.lhomChannel.changeIdBranch},
if $\id_\cM \otimes \eta_2 : H_f \xto{\bm\beta} G_2$ is an edge-bijective \ac{lhc},
then $\id_{\cX_1} \otimes \eta_2 : G_1 \xto{\bm\beta} F$ is an edge-bijective \ac{lhc}.
Hence, if we can construct encoders
\begin{align}
  \eta_1 \otimes \id_\cM &: H_f \xto{\bm\alpha} G_1
  \,,
                         &
  \id_\cM \otimes \eta_2 &: H_f \xto{\bm\beta} G_2
  \,,
\end{align}
we obtain an encoder $\Enc' = \eta_1 \otimes \eta_2$
such that
\begin{gather}
  \Dec \compose \phi \compose \Enc' : H_f \xto{\bm\alpha + \bm\beta + \bm\mu} \set{0,1}^1
\end{gather}
is an edge-bijective \ac{lhc}, and thus by
\cref{lemma:fcodeMakesLhomChannel.homToCode},
we have the following theorem.
In general, such a decoder exists by
\cref{lemma:fcodeMakesLhomChannel.homToCode} if $\phi : F \xto{\bm\mu} D$
is an edge-bijective \ac{lhc}, for some partition hypergraph $D$ with $\card{\cE(D)} = 2$.

\begin{theorem}
  \label{thm:idCode.fromIndependentEnc}
  For all hypergraphs $H = H_{f_{\textsf{ID}}}, G_1 = I, G_2 = G, F$
  as in \cref{lemma:bipartite.lhomChannel.changeIdBranch}
  and any $D$ as above with $\card{\cE(D)} = 2$, if
  \begin{align}
    \Enc_1 \otimes \id_\cM &: H \xto{\bm\alpha} G_1
    \,, \\
    \id_\cM \otimes \Enc_2 &: H \xto{\bm\beta} G_2
    \,, \\
    \phi &: F \xto{\mu} D
  \end{align}
  are edge-bijective \acp{lhc},
  then there exists a edge-bijective homomorphism $\Dec : D \to \set{0,1}^1$
  such that $(\Enc_1 \otimes \Enc_2, \Dec)$ is an
  $(M, \phi, \bm\alpha + \bm\beta + \bm\mu)$-\ac{id}-code.
\end{theorem}
\begin{remark}
This shows that knowing $F$, $G_1$, and $G_2$, one can design both encoders
as if the other message were known to the receiver.
The challenge is, however, finding a suitable $F$.
In the next example there exists a suitable $F$, based on a traditional
deterministic \ac{id} code where one message is known at the receiver.
It is not clear whether a suitable $F$ exists for a discrete memoryless
multiple-access channel
and super-exponential $\card\cM$.
\end{remark}

\begin{example}[Deterministic \ac{id}]

To illustrate \cref{thm:idCode.fromIndependentEnc},
consider \ac{id} with a deterministic encoder
over a memoryless channel $\phi^n = \BSC_\gamma^n \otimes \BSC_\gamma^n$, where
$\Pr\set{ y^n = \BSC_\gamma^n(x^n) } = \prod_{i=1}^n \Pr\set{ y_i = \BSC_\gamma(x_i) }$,
$\Pr\set{ x = \BSC_\gamma(x) } = 1 - \gamma \ne \frac 1 2$.
If one message in $\cM = \set{ 1,\dots,M }$
is known to the receiver,
a reliable codebook is any set $\set{ x^n_m : m \in \cM }$
where the minimal Hamming distance
$\min_{m' \ne m} d_H(x^n_m, x^n_{m'}) = n \delta$,
for arbitrary $0 < \delta < 1$~%
\cite{jaja1985identification,salariseddighPeregBocheDeppe2022det_ID_powerConstraints_tit}.
In the case where both messages are sent over the channel $\phi^n$,
consider the following. For two words $x^n, \bar x^n$
with distance $n \delta = d_H(x^n, \bar x^n) = \card{\set{ i : x_i \ne \bar x_i }}$,
let $(Y^n, \bar Y^n) = \phi^n(x^n, \bar x^n)$. Then,
\begin{align}
  n \theta_\delta &= \expect \brack{d_H(Y^n, \bar Y^n)}
  \onecol\tab = \sum\nolimits_{i=1}^n \expect \ind{ Y_i \ne \bar Y_i }
  \\&= \sum_{i \::\: x_i = \bar x_i}
         2 \gamma(1-\gamma)
     + \sum_{i \::\: x_i \ne \bar x_i}
         \tup{\gamma^2 + (1 - \gamma)^2}
  \\&= n \tup{1 - \delta} \cdot 2 \gamma (1 - \gamma)
    + n \delta \cdot \tup{ 1 - 2\gamma(1 - \gamma) }
  \\&= n \beta + n \delta \tup{ 1 - 2 \beta }
    \,,
\end{align}
where $\beta = 2\gamma(1-\gamma)$ is the probability that exactly one
letter of $x_i$ and $\bar x_i$ is flipped, and $\beta \le \frac 1 2$.
Consider now
\begin{align}
  0 < &\:\epsilon
  < \frac { \theta_\delta - \theta_0 } { \theta_\delta + \theta_0 }
  = \frac { \delta (1-2\beta) } { 2\beta + \delta (1-2\beta) }
  \le 1 - 2 \beta
  \,, \\
  \cX_\delta &= \set{ (x^n, \bar x^n) : d_H(x^n, \bar x^n) \ge n \delta }
  \,, \\
  \cX_0 &= \set{ (x^n, \bar x^n) : x^n = \bar x^n }
  \,, \\
  T_{\delta,\epsilon}
  &= \set{ (y^n, \bar y^n) :
    \card{ d_H(y^n, \bar y^n) - n \theta_\delta } < n \epsilon \theta_\delta
  }
  \,,
\end{align}
and a codebook $\cC = \set{ x_m^n : m \in \cM } \subseteq \cX_\delta$
as described in~\cite{jaja1985identification}.
We define the hypergraphs
\begin{align}
  H &\coloneqq
  H_{f_{\textsf{ID}}}
  = \tup[\big]{ \cM^2,
  \set[\big]{ \set{ (m,m) }, \set{ (m, \bar m) }_{m \ne \bar m} }
  }
  \,, \\
  G_1 &\coloneqq \tup[\big]{ \cX^n \times \cM,
    \set[\big]{
      \set{ (x^n_m, m) },
      \set{ (x^n_m, \bar m) }_{m \ne \bar m}
    }
  }
  \,, \\
  G_2 &\coloneqq \tup[\big]{ \cM \times \cX^n,
    \set[\big]{
      \set{ (m, x^n_m) },
      \set{ (m, x^n_{\bar m}) }_{m \ne \bar m}
    }
  }
  \,, \\
  F &\coloneqq \tup{ \cX^{2n}, \set{ \cX_0\,,\, \cX_\delta } }
  \,, \\
  C &\coloneqq F \cap \cC^2 \coloneqq \tup{ \cC^2, \set{ \cX_0 \cap \cC^2 \,,\, \cX_\delta \cap \cC^2 } }
  \,, \\
  D &\coloneqq \tup{ \cX^{2n}, \set{ T_{0, \epsilon}\,,\, T_{\delta, \epsilon} } }
  \,.
\end{align}
By Chernoff's bound~\cite[Theorems 4.4 and 4.5]{mitzenmacherUpfal2007probabilityComputing}
$
  \Pr\set{ (Y^n, \bar Y^n) \notin T_{\delta,\epsilon} }
  \le 2 e^{ - n \epsilon^2 \theta_\delta / 2 }
  $%
  ,
which converges to zero for increasing $n$, for all $0 \le \delta \le 1$.
Hence, for $\Enc_1 = \Enc_2 : m \mapsto x^n_m$,
$\Enc_1 \otimes \id_\cM : H \to G_1$, $\id_\cM \otimes \Enc_2 : H \to G_2$,
and $\Enc_1 \otimes \Enc_2 : H_f \to C$
are edge-bijective homomorphisms, and whenever $0 < \bm\lambda < 1$,
$\phi^n : F \xto{\bm\lambda} D$ is an edge-bijective \ac{lhc}
with edge map $f_\cE : \cX_\alpha \mapsto T_{\alpha,\epsilon}$, $\alpha = 0,\delta$.
This confirms \cref{lemma:bipartite.lhomChannel.changeIdBranch},
and by \cref{thm:idCode.fromIndependentEnc}, $(\Enc_1 \otimes \Enc_2, \Dec)$
is an $(M, \phi^n, \bm\lambda)$-\ac{id}-code.
Note that \cref{lemma:bipartite.lhomChannel.changeIdBranch} assumes
that $F$ is a partition hypergraph $F$ with a rectangular vertex set.
Our $F$ is no partition hypergraph, and $\cX_0 \bigcup \cX_\delta$ is not
rectangular. However, $C$, which is $F$ restricted to the codebook $\cC$,
satsifies both.
\end{example}

\cref{fig:idMacBipartiteRates} visualizes the known achievable rates
for our \ac{id} model and the capacity regions for traditional \ac{id} and
transmission over over $\BSC_\gamma^n \otimes \BSC_\gamma^n$,
when $\gamma = 0.03$ and only determinstic encoding is allowed.

\begin{figure}
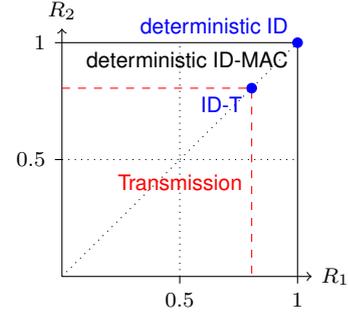

  \centering
  \tikzIdMacBipartiteRates[
    x=0.35\linewidth,y=0.35\linewidth,
    rateBound/di/node/.style = above left,
    rateBound/di-t/node/.style = below left,
    rateBound/di-mac/node/.style = below left,
    rateBound/transmission/node/.style = left,
  ]{0.03}
  \caption{\label{fig:idMacBipartiteRates}
    \ac{id} and transmission rates achievable over two parallel binary symmetric
    channels with crossover probability $\gamma = 0.03$, using deterministic encoders.
    \textsf{\acs{id}} corresponds to our $f_{\textsf{ID}}$, where the two
    messages to be compared are sent over the channel, whereas \textsf{ID-MAC}
    corresponds to $f_{\textsf{ID-MAC}} = f_{\textsf{ID}}^{(1)} \otimes f_{\textsf{ID}}^{(2)}$,
    where $f_{ID}^{(i)} : \cM_i^2 \to \set{0,1}$, are \ac{id} functions,
    i.e., two simultaneous \acp{id} are performed, each with one message known
    to the receiver.
    For deterministic \textsf{\acs{id}}, the diagonal up to the marked points is
    achievable. For deterministic \textsf{\ac{id}-MAC}, the regions below the lines are
    achievable~\cite[Sec.~4.A]{rosenbergerIbrahimDeppeFerrara2023di_mac_isit}.
    For transmission over a MAC, the red dashed line is the boundary of the
    capacity region~\cite{elgamalKim2011network_it}.
    The intersection point \textsf{\acs{id}-T} marks the rates achievable for
    $f_{\textsf{ID}}$ if we restrict the codes to transmission codes.
  }
\end{figure}

\section{Conclusion}

We developed the notion of \iacf{lhc}
and proved approximate equivalence between those and codes for computing
functions. Further, we derived decomposition properties of \ac{lhc}
which used to analyze and construct codes where two messages must be encoded
independently. This lead to new results for \acf{id} and $K$-\ac{id},
which were illustrated at the example of \ac{id} with deterministic encoders.
Despite surprising improvements compared to naive code constructions,
in many cases it remains unknown how optimal our constructions are.
To assess this, one needs to enhance understanding of
the input hypergraphs to \acp{lhc}. In particular, it is still unclear
if the vertex set of the target hypergraph of a bipartite encoder must be
strictly rectangular, and how to optimize these, which are also the input
hypergraphs to the actual channel, such that the vertex sets are
rectangular and the channels are locally homomorphic and edge-bijective.
Answering these questions would fully explain the tradeoff between the
rates of the two messages.
A first application is presented in~\cite{rosenbergerBocheCabreraFitzek2024consensus_globecom}.

\section*{Acknowledgement}

The authors acknowledge the financial support by the Federal Ministry of
Education and Research of Germany (BMBF) in the program of “Souverän. Digital.
Vernetzt.” Joint project 6G-life, project identification numbers:
16KISK001K, 16KISK002.
H. Boche was further supported in part by the BMBF within the
national initiative on Post Shannon Communication (NewCom) under Grant
16KIS1003K. C. Deppe also received support from
the BMBF within the national initiative on Post Shannon Communication
(NewCom) under Grant 16KIS1005.
J. Rosenberger and C. Deppe were also supported
by the German Research Foundation (DFG) in the German-Israeli Project Cooperation (DIP), under grant
number 509917421, and C. Deppe received support from the DFG in the project DE1915/2-1.
Moreover, C. Deppe was supported by the Bavarian Ministry of Economoic Affairs,
Regional Development and Energy in the project 6G and Quantum Technologies (6G QT).
J. Cabrera received additional support from the German Research Foundation (DFG,
Deutsche Forschungsgemeinschaft) as part of Germany’s Excellence Strategy—EXC
2050/1—Cluster of Excellence ``Centre for Tactile Internet with Human-in-the
loop'' (CeTI) of Technische Universitat Dresden under Project ID: 390696704.

\fxwarning[author=JR]{other funding}

\printbibliography

\newpage
\listoffixmes

\end{document}